\documentclass[12pt]{article}
\usepackage{amsmath,amssymb,mathpazo,fullpage,microtype}
\usepackage{amsthm}

\usepackage{color}
\usepackage{enumitem}
\usepackage{textcomp}

\usepackage[titletoc,title]{appendix}

\definecolor{webgreen}{rgb}{0,.35,0}
\definecolor{webbrown}{rgb}{.6,0,0}
\definecolor{RoyalBlue}{rgb}{0,0,0.9}

\usepackage{hyperref}
\usepackage{cleveref}
\crefformat{equation}{(#2#1#3)}
\crefrangeformat{equation}{(#3#1#4) to~(#5#2#6)}
\crefmultiformat{equation}{(#2#1#3)}{ and~(#2#1#3)}{, (#2#1#3)}{ and~(#2#1#3)}
\usepackage{bm}
\hypersetup{
   colorlinks=true, linktocpage=true, pdfstartpage=3, pdfstartview=FitV,
   breaklinks=true, pdfpagemode=UseNone, pageanchor=true, pdfpagemode=UseOutlines,
   plainpages=false, bookmarksnumbered, bookmarksopen=true, bookmarksopenlevel=1,
   hypertexnames=true, pdfhighlight=/O,
   urlcolor=webbrown, linkcolor=RoyalBlue, citecolor=webgreen,
   pdfauthor={Thomas Fai and Youngmin Park},
   pdftitle={Global asymptotic stability of the active disassembly model of flagellar length control},
   pdfkeywords={},
   pdfcreator={pdfLaTeX},
   pdfproducer={LaTeX with hyperref}
}

\usepackage{graphicx}
\graphicspath{{figures/}{.}}

\SetLabelAlign{LeftAlignWithIndent}{\hspace*{\the\dimexpr1.31in+\oddsidemargin\relax}\makebox[10.25em][l]{#1}}

\newcommand{\ve}{\varepsilon}

\newtheorem{definition}{Definition}
\newtheorem{theorem}{Theorem}

\newtheorem*{claim*}{Claim}

\newcommand{\ud}{\mathrm{d}}
\newcommand{\oh}{\frac{1}{2}}

\newcommand{\abs}[1]{\left\lvert #1 \right\rvert}

\usepackage{graphicx}
\usepackage{epstopdf}
\usepackage{subfig}

\providecommand{\keywords}[1]{\textbf{\textit{Key words:}} #1}

\usepackage{authblk}
\author[$\star,\dagger$]{Thomas G. Fai}
\author[$\star$]{Youngmin Park}
\affil[$\star$]{{\small Department of Mathematics, Brandeis University, Waltham, Massachusetts 02453}}
\affil[$\dagger$]{{\small Volen Center for Complex Systems, Brandeis University, Waltham, Massachusetts 02453}}

\title{Global asymptotic stability of the active disassembly model of flagellar length control}
\date{}	

\begin{document}
\maketitle

\begin{abstract}
Organelle size control is a fundamental
question in biology that demonstrates the fascinating ability of cells to maintain homeostasis within
their highly variable environments. Theoretical models describing cellular dynamics have the potential to help elucidate the principles underlying size control.
Here, we perform a detailed study of the active disassembly model proposed in [Fai et al, Length regulation of multiple flagella that self-assemble from a shared pool of components, \emph{eLife}, 8, (2019): e42599]. {\color{black} We construct} a hybrid system which is shown to be well-behaved throughout the domain. We rule out the possibility of oscillations arising in the model and prove global asymptotic stability in the case of two flagella by the construction of a suitable Lyapunov function. {\color{black} Finally, we generalize the model to the case of arbitrary flagellar number in order to study olfactory sensory neurons, which have up to twenty cilia per cell.} We show that our theoretical results may be extended to this case and explore the implications of this universal mechanism of size control.
\end{abstract}

\vspace{.2in}

\noindent
Robust size control is essential in biology at the levels of organs, cells, and organelles. The absence of size regulation is a marker of pathologies such as cancer \cite{zink2004nuclear}, and the varied sizes of biological structures highlight differences in the fundamental physical constraints on biological systems \cite{haldane1926being,marshall2016cell,milo2015cell}.

At the level of organelles, a particular useful model organism for flagellar length control is the unicellular green algae \emph{Chlamydomonas reinhardtii} (Figure \ref{fig:schematic}(a))\cite{wemmer2007flagellar,gutman2004chlamydomonas,goldstein2015green,salome2019series,wan2020unity}. \emph{Chlamydomonas} uses its two flagella to propel itself through its aqueous environment and the lengths of its flagella are important for maintaining swimming speed and direction \cite{nguyen2005lf1,khona2013anomalies}.

{\color{black}
Each flagellum is made up of a microtubule structure known as the axoneme, in which nine microtubule doublets surround a central pair (referred to as the $9+2$ structure) \cite{morga2013getting}. In flagellates such as \emph{Chlamydomonas}, the sliding motion between these microtubule doublets generated by the molecular motor dynein leads to an emergent beating along the flagellum.

The internal structures of flagella are highly dynamic and in continual turnover. During intraflagellar transport (IFT), protein assemblies called IFT particles are transported by motor molecules, some of which carry particles toward the tip of the flagella (anterograde motion) and others that carry IFT particles toward the base of the flagella (retrograde motion) \cite{kozminski1993motility}. These IFT particles move processively with speeds that are approximately constant. There is a difference in speed between anterograde (2 \textmu m/s) and retrograde (3.5 \textmu m/s)  transport is explained by the fact that different species of molecular motors---kinesin and dynein---are responsible for walking along the microtubule filaments in the anterograde and retrograde directions, respectively \cite{scholey2003intraflagellar}.
}

Experiments on \emph{Chlamydomonas} have shown that flagellar maintenance and length control is a dynamic process, which is closely linked to protein transport by IFT. {\color{black} For example, researchers have performed severing experiments in which one of the two flagella is removed from the cell and observed to regenerate. This severing may be achieved experimentally through either mechanical shearing \cite{rosenbaum1969flagellar} or laser ablation \cite{ludington2012organelle}. After severing, the shortened flagellum regrows to approach a new and somewhat smaller steady-state length. Interestingly, the unsevered flagellum also responds; it shortens until it reaches approximately the same length as the severed flagellum, revealing the coupling between the two flagella. After this relatively rapid initial phase (order of seconds), there is a subsequent phase on a slower timescale (order of minutes)} in which the two flagella grow back in unison to the initial steady-state length (shown in Figure \ref{fig:schematic}(b)).
\begin{figure}[]
\centering
	 \vspace*{-.4in}
	 \hspace*{-.15in}\includegraphics[scale=0.9]{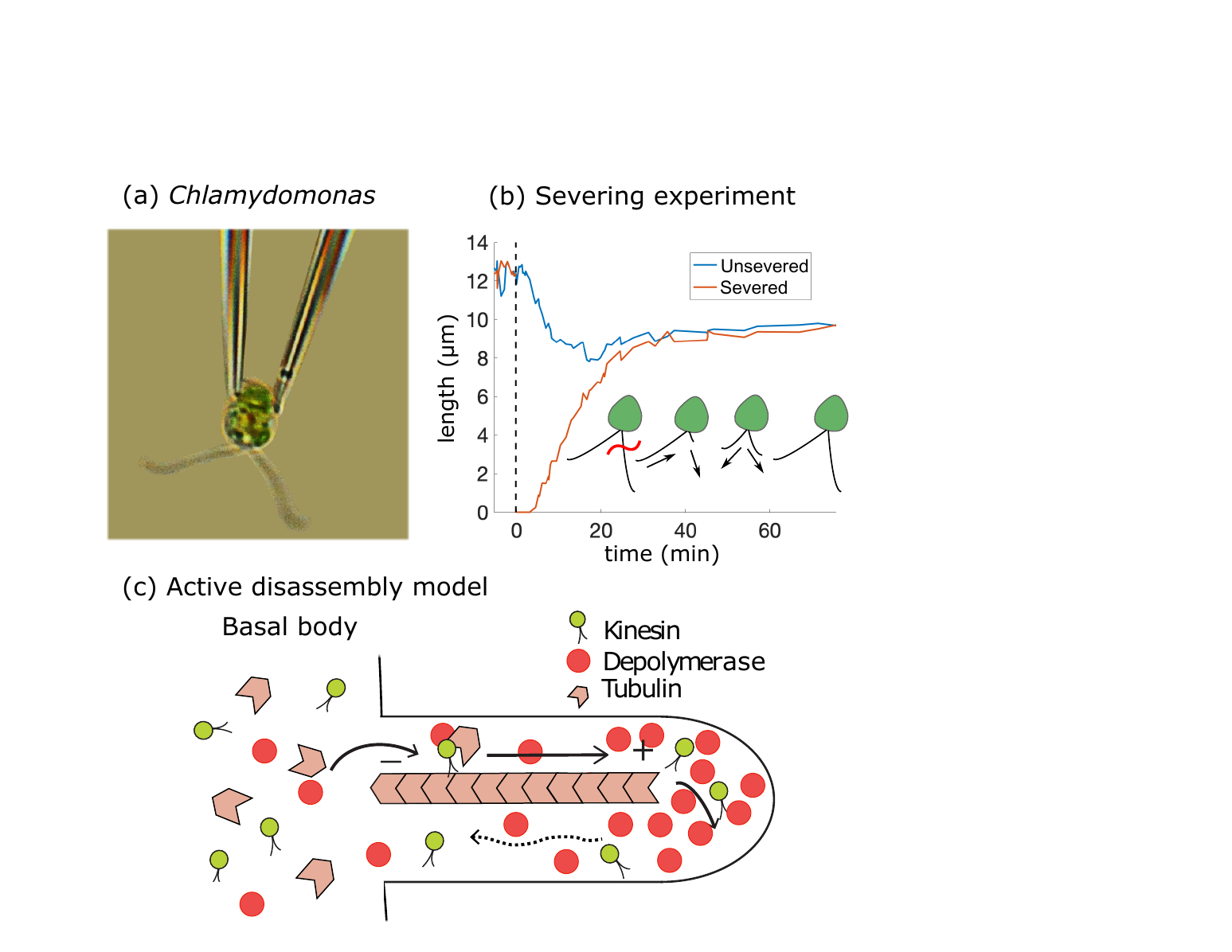}
	 \vspace*{-.4in}
	 \caption{\emph{Chlamydomonas reinhardtii} and the active disassembly model of flagellar length control. (a) Experimental image of \emph{Chlamydomonas} from \cite{goldstein2015green} provided courtesy of the author. The flagella have lengths of approximately 12 \textmu m. {\color{black} (b)} Upon severing one flagellum, the severed flagellum begins to regenerate while the unsevered flagellum shortens. The flagella rapidly equalize at an intermediate length before returning to their initial lengths on a slower timescale. {\color{black} Data from \cite{ludington2012organelle} provided by the authors}, (c) Schematic of the active disassembly model, in which the concentration gradient generated by a diffusing depolymerase leads to a length-dependent rate of disassembly at the flagellar tip.}
         \label{fig:schematic}
\end{figure}

In \cite{fai2019length}, we proposed an active disassembly model for flagellar length regulation in \emph{Chlamydomonas} based on physical first principles of flux balance (Figure \ref{fig:schematic}(c)). By providing a mechanism for length-dependent disassembly, this model leads to simultaneous length control consistent with the length-equalization observed in severing experiments.

In this work, we consider the mathematical properties of the system of nonlinear ordinary differential equations that make up the active disassembly model. While it is well-known that in principle nonlinear equations may lead to behaviors such as finite time blow-up or oscillations, we show that these behaviors can be ruled out over a wide range of parameters.
{\color{black}
Moreover, we generalize the model to the case of arbitrary flagellar number, motivated by the olfactory sensory neurons which have on the order of 5--20 cilia per cell.}

We show that these equations have a global asymptotically stable steady-state solution. To prove this, we find a Lyapunov function for the dynamics. The Lyapunov function also implies there are no limit cycles or periodic oscillations. Further, we construct a hybrid system to ensure that the flagellar lengths remain non-negative, and this yields regions of state space from which all trajectories contract to a single point on the boundary before flowing toward steady-state. We interpret these trajectories in light of existing experimental data. In addition, we provide alternate proofs of global asymptotic stability and the non-existence of periodic oscillations that do not rely on the construction of a Lyapunov function. We discuss how these results generalize to other situations.

By characterizing the mathematical properties of the active disassembly model, we aim to provide a strong foundation for subsequent applications and analysis of the model. In the Discussion, we note several open questions regarding the mathematical model and its application to organelle size control.

\section{Active disassembly model of flagellar length control}
\label{sec:model}
{\color{black}
As described in the Introduction, the process of IFT involves the transport of various biomolecules between the flagellar base and tip, and is believed to be closely linked with the dynamics of flagellar length. In the development of the mathematical model described next, we conceptually simplify this process by capturing only the essential biomolecular constituents for length control. In particular, we restrict attention to the tubulin that makes up the axoneme structure, the molecular motors that power transport along the axoneme, and the depolymerase conjectured to be necessary for length control. We assume that both assembly and disassembly are in continual competition so that the overall growth rate is the net result of these processes. Next, we describe the assembly and disassembly rates that give rise to differential equations for the flagellar lengths.
}

Let $L_1(t)$ and $L_2(t)$ denote the lengths of two flagella on a single \emph{Chlamydomonas} organism at a particular time $t$. The flagellum is composed of tubulin filaments, and we assume there is a total amount $T$ of tubulin available. {\color{black} The units of tubulin are given in microns, so that the pool size and flagellar lengths have the same units and the amount $T_f$ of free tubulin at the flagellar base is given by $T_f = T-L_1-L_2$.}

We assume there are a total number $M$ of anterograde molecular motors, of which a number $M_f$ are free in the basal pool while the others are undergoing IFT. In line with the pattern of motion found experimentally for kinesin \cite{chien2017dynamics}, the motors move ballistically in the anterograde direction {\color{black} with velocity $v$} and diffusively in the retrograde direction with diffusion coefficient $D$.

The injection flux $J$ of IFT particles into the flagella satisfies mass action kinetics, so that
 \begin{equation}
J=k_\text{on}M_f. \label{eq:flux0_main}
\end{equation}
There is a separation of timescales between the motion of IFT particles, which traverse the flagellum in several seconds, and the flagellar length dynamics, which varies over a timescale of several minutes \cite{hendel2018diffusion}. {\color{black} As shown in Appendix \ref{app:flux}}, upon combining \eqref{eq:flux0_main} with conservation of protein number, this separation of timescales leads to a quasi-steady state approximation for the flux:
\begin{equation}
J = \frac{k_\text{on}M/2}{1+k_\text{on}(L_1+L_2)/2v+k_\text{on}(L_1^2+L_2^2)/4D}.
\end{equation}

The flux $J$ represents the number of IFT particles injected per second into the flagella, and these particles are assumed to carry an amount of tubulin proportional to the free tubulin in the basal pool. Therefore, the overall speed of assembly is equal to $\gamma J \left(T-L_1-L_2\right)$, where $\gamma$ is a proportionality constant that represents the fraction of tubulin arriving at the flagellar tip that is incorporated into the flagellum.

Disassembly is taken to be a linear function of the depolymerase concentration at the flagellar tip, i.e.~each depolymerizing enzyme located near the flagellar tip removes a constant number of tubulin subunits per second. This implies that for the $i$\textsuperscript{th} flagellum for $i=1,\,2$, the speed of disassembly is equal to $d_0+d_1 \frac{J L_i}{D}$, {\color{black} where $d_0$ and $d_1$ are the coefficients of the first-order Taylor series approximation. As explained in detail in Appendix \ref{app:flux}, this follows from the assumption that there is a depolymerase which has the same ballistic-to-diffusive pattern of motion as kinesin: the disassembly rate equals $d_0+d_1 c_d(L_i)$, where $c_d(L_i)$ is the concentration of depolymerase at the flagellar tip.}

Setting the growth rate to be the difference between assembly and disassembly speeds leads to a system of coupled nonlinear differential equations:
\begin{align}
L_1' &=  \gamma J(T-L_1 - L_2) - d_0 - d_1 \frac{JL_1}{D} \label{eq:active1_det}\\
L_2' &=  \gamma J(T-L_1 - L_2) - d_0 - d_1 \frac{JL_2}{D}. \label{eq:active2_det}
\end{align}
where $J\equiv J(L_1,L_2)$ satisfies
\begin{equation}\label{eq:J}
J(L_1,L_2) = \frac{k_\text{on} M/2}{1 + k_\text{on} (L_1+L_2)/2v + k_\text{on} (L_1^2 + L_2^2)/4D}.   
\end{equation}

\subsection{Nondimensionalization}
We next rewrite \cref{eq:active1_det,eq:active2_det} in dimensionless form. In terms of the nondimensional length $\widetilde{L}=L/T$ and nondimensional time $\widetilde{t}= t \gamma k_\text{on} M/2$ as well as the dimensionless parameters $\pi_0=2d_0/\gamma k_\text{on} MT$, $\pi_1=d_1/\gamma D$, $\pi_2=k_\text{on}T/2v$, and $\pi_3=k_\text{on}T^2/4D$, the system of equations \cref{eq:active1_det,eq:active2_det} becomes
\begin{align}
\frac{\ud \widetilde{L}_1}{\ud \widetilde{t}}=\widetilde{J}\left(1-\widetilde{L}_1-\widetilde{L}_2\right)-\pi_0-\pi_1\widetilde{J}\widetilde{L}_1 \label{eq:active1_nondim}\\
\frac{\ud \widetilde{L}_2}{\ud \widetilde{t}}=\widetilde{J}\left(1-\widetilde{L}_1-\widetilde{L}_2\right)-\pi_0-\pi_1\widetilde{J}\widetilde{L}_2,\label{eq:active2_nondim}
\end{align}
where the dimensionless flux $\widetilde{J}=2J/k_\text{on} M$ satisfies
\begin{equation}\label{eq:J_nondim}
\widetilde{J}=\frac{1}{1+\pi_2(\widetilde{L}_1+\widetilde{L}_2)+\pi_3(\widetilde{L}_1^2+\widetilde{L}_2^2)}.
\end{equation}
We interpret the dimensionless parameters as follows: $\pi_0$ is the ratio of disassembly and assembly rates, $\pi_1$ is the depolymerase activity level, $\pi_2=\tau_b/\tau_i$ is the ratio of the ballistic timescale $\tau_b:=T/v$ of IFT transport to the injection timescale $\tau_i:=k_\text{on}^{-1}$, and $\pi_3=\tau_d/\tau_i$ is an analogous ratio of the diffusive timescale $\tau_d=T^2/2D$ to the injection timescale. (We could equivalently think of $\pi_2$ and $\pi_3$ as ratios of lengthscales related to the same physical processes.)

It is a straightforward calculation to verify that both flagella have the same steady-state length $\widetilde{L}_\text{ss}$ given by
\begin{equation}
\label{eq:Lss}
\widetilde{L}_\text{ss} = \frac{1+\pi_0 \pi_2 + \pi_1/2}{2\pi_0 \pi_3}\left( \sqrt{1+\frac{2\left(1-\pi_0\right)\pi_0 \pi_3}{\left(1+\pi_0 \pi_2+\pi_1/2\right)^2}}-1\right).
\end{equation} 

\subsection{Generalization to arbitrary flagellar number}
Here we generalize the flagellar length control model to arbitrary $N$. We are motivated by cells such as olfactory sensory neurons (OSN's), which have on the order of a dozen cells, as we shall discuss later on in detail. In this case, the equations corresponding to \cref{eq:active1_det,eq:active2_det,eq:J} become:
\begin{align}
\frac{\ud L_i}{\ud t} &= \gamma J\left(T-\sum_{j=1}^N L_j\right)-d_0-d_1\frac{JL_i}{D}, \quad i=1,\dots,N,\\
J(L_1,\dots,L_N) &= \frac{k_\text{on} M/N}{1 + \left(k_\text{on}/Nv\right)\sum_{j=1}^N L_j+ \left(k_\text{on}/2ND\right)\sum_{j=1}^N L_j^2}.
\end{align}

To put the equations in dimensionless form, we follow the approach taken in the case $N=2$ above by introducing the dimensionless variables $\widetilde{L}_i=L_i/T$ and $\widetilde{t}= t \gamma k_\text{on} M/N$ as well as the dimensionless parameters $\pi_0=Nd_0/\gamma k_\text{on} MT$, $\pi_1=d_1/\gamma D$, $\pi_2=k_\text{on}T/Nv$, and $\pi_3=k_\text{on}T^2/ND$. This leads to the following system, which generalizes \cref{eq:active1_nondim,eq:active2_nondim,eq:J_nondim} above:
\begin{align}
\frac{\ud \widetilde{L}_i}{\ud \widetilde{t}}&=\widetilde{J}\left(1-\sum_{j=1}^N \widetilde{L}_j\right)-\pi_0-\pi_1\widetilde{J}\widetilde{L}_i, \quad i=1,\dots,N, \label{eq:activeN_nondim}\\
\widetilde{J}&=\frac{1}{1+\pi_2 \sum_{j=1}^N \widetilde{L}_j+\pi_3\sum_{j=1}^N \widetilde{L}_j^2}.
\end{align}
As in the case $N=2$, \eqref{eq:activeN_nondim} leads to a unique fixed point with $\widetilde{L}_i=\widetilde{L}_\text{ss}$ for $i=1,\dots,N$ with
\begin{equation}
\label{eq:LssN}
\widetilde{L}_\text{ss} = \frac{1+\pi_0 \pi_2 + \pi_1/N}{2\pi_0 \pi_3}\left( \sqrt{1+\frac{4\left(1-\pi_0\right)\pi_0 \pi_3}{N\left(1+\pi_0 \pi_2+\pi_1/N\right)^2}}-1\right),
\end{equation}
which decays as $1/N$ as $N\to\infty$.

Note that, in the limit $N\to\infty$ for fixed $T$ and $M$, the dimensionless basal disassembly parameter $\pi_0$ satisfies $\pi_0 \to \infty$. This is because the injection rate into any one flagellum, which goes as $1/N$, becomes smaller and smaller in comparison to the $N$-independent basal disassembly $d_0$, resulting in steady-state lengths of zero beyond a critical value of $N$. However, in the special case $d_0=0$, one obtains flagella of progressively smaller but nonzero lengths as $N\to 0$.

{\color{black}
\subsection{Ofactory sensory neuron cilia}
We next apply the active disassembly model to data on olfactory sensory neurons (OSN's) from the mouse nasal septum. These cells each have on the order of 5--20 cilia, and the number depends on the cell's location the within the dorsal zone. (The structures of cilia and flagella are the same and the terms may be used interchangeably.) In particular, it is reported in \cite{challis2015olfactory} that cells in the anterior region have $N_a = 14 \pm 2.9$ cilia, whereas cells in the posterior region have $N_p = 5.6 \pm 1.8$ cilia. Not only the cilia number but also the average lengths are location-dependent, with anterior cilia having lengths of $L_a = 14.7 \pm 10.7$ \textmu m and posterior cilia having shorter lengths of $L_p  =2.8 \pm 2.1$ \textmu m.
 \begin{figure}[]
\centering
	 \hspace*{.45in}\includegraphics[scale=1,trim = 1.5in 3.5in 1in 1.5in, clip]{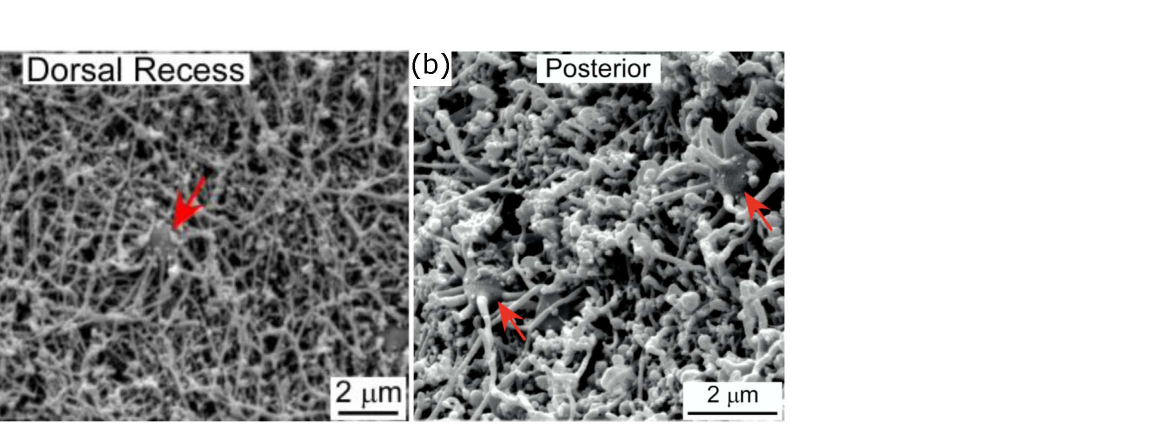}
         \caption{Surface SEM images from the dorsal recess and posterior mouse nasal septum adapted from \cite{challis2015olfactory} (a) Examination of the highlighted dorsal recess OSN yields an estimated diameter of 1.2 \textmu m, (b) Posterior OSN's are estimated to have a diameter of 0.93 \textmu m. Images reprinted from \cite{challis2015olfactory} (Figs.~1(D) and S2(D)) with permission from Elsevier.}
         \label{fig:OSN}
\end{figure}

This data provides an opportunity to apply our model of length control to study the correlation between cilia number and length. Notably, it is \emph{not} the case that posterior cells having fewer cilia tend to have longer cilia. Interpreting this data through the active disassembly model, which in the limit of large $N$ implies the steady-state lengths satisfy $L \approx T/N$, this implies that cells in different regions have different tubulin pool sizes. Accordingly, we estimate the pool sizes of OSN's in the anterior to be $T_a \approx N_a \cdot L_{a} = 225$ \textmu m and for OSN's in the posterior to be $T_p \approx N_p \cdot L_{p} = 18$ \textmu m.

If we assume protein concentrations are held roughly constant across the dorsal zone, this means the pool size should scale with volume, and based on the spherical shape of OSN's implies a scaling of $T \propto R^3$. Plugging in the numbers above, the model predicts that $R_a/R_p = \sqrt[3]{225/18} \approx 2.3$, whereas direct measurement of the cell radii from published images of OSN's yields a ratio of $R_a/R_p \approx 1.3$. As data is sparse, we use an available image from the dorsal recess to approximate the diameter of anterior OSN's. See Figure \ref{fig:OSN}, which shows how the OSN diameters are estimated using images from \cite{challis2015olfactory}.

This analysis leads to several conclusions. First, the assumption of constant concentrations of ciliary proteins appears to be inconsistent with the variations of cilia length and number observed across the dorsal zone. This indicates that concentration gradients may be responsible for sustaining the different pool sizes predicted by the model. Interestingly, computing the ratio of cilia number yields $N_a/N_p \approx 2.3$, which suggests that cilia number scales with volume. The mechanism for this apparent scaling relation remains an open question.
}
\section{Nonlinear behavior}
Next, we turn to the mathematical analysis of the dynamical equations.
To simplify notation, we drop tildes and use the dimensionless formulation for the remainder of the article. We will first consider the case of $N=2$ flagella, and later on show that the key results generalize to arbitrary flagellar number.
The active disassembly model \cref{eq:active1_nondim,eq:active2_nondim} is not a gradient system, i.e.~it cannot be written as the gradient of an energy function $E(L_1,L_2)$. This is proven by contradiction. Suppose there is an energy function such that $L_1' = \partial E/\partial L_1$ and $L_2' = \partial E/\partial L_2$. Then, by equality of mixed partials
\begin{equation}
 \frac{\partial^2 E}{\partial L_1\partial L_2} = \frac{\partial L_1'}{\partial L_2} = \frac{\partial L_2'}{\partial L_1} = \frac{\partial^2 E}{\partial L_2\partial L_1},
\end{equation}
i.e.\
\begin{equation}
\frac{\partial \left(J\left(1-L_1-L_2\right)-\pi_0-\pi_1JL_1\right)}{\partial L_2} = \frac{\partial \left(J\left(1-L_1-L_2\right)-\pi_0-\pi_1JL_2\right)}{\partial L_1}.
\end{equation}
However, the left-hand side of the {\color{black} equality} is equal to
\begin{equation}
-J+\left(1-L_1-L_2\right)\frac{\partial J}{\partial L_2}-\pi_1 L_1 \frac{\partial J}{\partial L_2},
\end{equation}
whereas the right-hand side is equal to
\begin{equation}
-J+\left(1-L_1-L_2\right)\frac{\partial J}{\partial L_1}-\pi_1 L_2 \frac{\partial J}{\partial L_1}.
\end{equation}
Because $J(L_1,L_2)=J(L_2,L_1)$, we have $\partial J/\partial L_1=\partial J/\partial L_2$. Assuming $\pi_1\neq 0$, {\color{black} it follows} that there is equality if and only if $L_1=L_2$. We have thereby reached the contradiction since $L_1 \neq L_2$ in general. It follows that no energy function exists.

However, using a technique applicable to many nonlinear dynamical systems, we may construct a Lyapunov function to show that the steady state is globally asymptotically stable. First, we specify the subspace $\Omega$ with $(L_1,L_2) \in \Omega \subset \mathbb{R}^2$ over which the model is defined. Because we must have $L_1,\, L_2 \ge 0$ for the variables to have physical meaning as lengths, $\Omega$ is contained within the first quadrant. Moreover, the sum of lengths cannot exceed the total pool size so that $L_1 + L_2 \le 1$. To summarize, $\Omega$ is defined to be the triangular region in the first quadrant bound by the lines $L_1 = 0$, $L_2 = 0$, and $L_1+L_2 \le 1$. 

We define the deviations from steady-state $\Delta L_1 := L_1-L_\text{ss}$ and $\Delta L_2 := L_2-L_\text{ss}$ on $\Omega$ and express the dynamics in terms of $\Delta L_1$ and $\Delta L_2$:
\begin{align}
\frac{\ud L_1}{\ud t} &= J\left(1-2L_\text{ss}-\Delta L_1-\Delta L_2\right)-\pi_0-\pi_1J\left(L_\text{ss}+\Delta L_1\right) \label{eq:active1_lyap1}\\
\frac{\ud L_2}{\ud t} &= J\left(1-2L_\text{ss}-\Delta L_1-\Delta L_2\right)-\pi_0-\pi_1J\left(L_\text{ss}+\Delta L_2\right) \label{eq:active2_lyap1},
\end{align}
or equivalently
\begin{align}
\frac{\ud L_1}{\ud t} &= J\left(1-2L_\text{ss}\right)-J\left(\Delta L_1+\Delta L_2\right)-\pi_0-\pi_1JL_\text{ss}-\pi_1J\Delta L_1
\label{eq:active1_lyap2}\\
\frac{\ud L_2}{\ud t} &= J\left(1-2L_\text{ss}\right)-J\left(\Delta L_1+\Delta L_2\right)-\pi_0-\pi_1JL_\text{ss}-\pi_1J\Delta L_2 \label{eq:active2_lyap2}.
\end{align}
By symmetry, the steady-state lengths $L_{1,\text{ss}}$ and $L_{2,\text{ss}}$ are equal. We denote the common steady-state length by $L_\text{ss}$. It is determined by setting the derivative equal to zero in \eqref{eq:active1_lyap2} or \eqref{eq:active2_lyap2}, which yields the identity
\begin{equation} \label{eq:active_ss}
0 = J_\text{ss}\left(1-2L_\text{ss}\right)-\pi_0-\pi_1J_\text{ss}L_\text{ss},
\end{equation}
where
\begin{equation}\label{eq:J_ss}
J_\text{ss}=\frac{1}{1+2\pi_2L_\text{ss}+2\pi_3L_\text{ss}^2}.
\end{equation}
(Eq.~\eqref{eq:active_ss} yields a quadratic equation for $L_\text{ss}$, which has a unique positive root given by \eqref{eq:Lss}.)

We next use the steady-state relationship \eqref{eq:active_ss} to rewrite \cref{eq:active1_lyap2,eq:active2_lyap2} in a different form. Before we proceed, we first note the identities given by
\begin{align}
 J\left(1-2L_\text{ss}\right)&=(J-J_\text{ss})\left(1-2L_\text{ss}\right)+J_\text{ss}\left(1-2L_\text{ss}\right) \label{eq:active1_interm1}\\
 \pi_1JL_\text{ss} &=  \pi_1(J-J_\text{ss})L_\text{ss}+ \pi_1J_\text{ss}L_\text{ss}, \label{eq:active2_interm1}
\end{align}
where the first line follows from adding and subtracting $J_\text{ss}\left(1-2L_\text{ss}\right)$ and the second line follows from adding and subtracting $\pi_1J_\text{ss}L_\text{ss}$. Plugging the above expressions into \cref{eq:active1_lyap2,eq:active2_lyap2} and using the steady-state relation \eqref{eq:active_ss} yields
\begin{align}
\frac{\ud L_1}{\ud t} &= \left(J-J_\text{ss}\right)\left(1-2L_\text{ss}\right)-J\left(\Delta L_1+\Delta L_2\right)-\pi_1\left(J-J_\text{ss}\right)L_\text{ss}-\pi_1J\Delta L_1
\label{eq:active1_lyap3}\\
\frac{\ud L_2}{\ud t} &= \left(J-J_\text{ss}\right)\left(1-2L_\text{ss}\right)-J\left(\Delta L_1+\Delta L_2\right)-\pi_1\left(J-J_\text{ss}\right)L_\text{ss}-\pi_1J\Delta L_2 \label{eq:active2_lyap3}.
\end{align}

Adding and subtracting the above equations yields
\begin{align}
\frac{\ud L_1}{\ud t}+\frac{\ud L_2}{\ud t} &= 2 \left(J-J_\text{ss}\right)\left(1-2L_\text{ss}-\pi_1 L_\text{ss}\right)-\left(2+\pi_1\right)J\left(\Delta L_1+\Delta L_2\right) \label{eq:active1_lyap12}\\
\frac{\ud L_1}{\ud t}-\frac{\ud L_2}{\ud t} &= -\pi_1J\left(\Delta L_1-\Delta L_2\right) \label{eq:active2_lyap12}
\end{align}

It follows that
\begin{align}
\oh \frac{\ud}{\ud t}\left( \Delta L_1+\Delta L_2\right)^2 &= 2 \left(J-J_\text{ss}\right)\left(1-2L_\text{ss}-\pi_1 L_\text{ss}\right)\left(\Delta L_1+\Delta L_2\right)-\left(2+\pi_1\right)J\left(\Delta L_1+\Delta L_2\right)^2  \label{eq:active1_interm2}\\
\oh \frac{\ud}{\ud t}\left( \Delta L_1-\Delta L_2\right)^2 &= -\pi_1J\left(\Delta L_1-\Delta L_2\right)^2. \label{eq:active2_interm2}
\end{align}
where we have used $\ud L_i/\ud t = \ud (L_\text{ss}+\Delta L_i)/ \ud t = \ud \Delta L_i/\Delta t$ for $i=1,\,2$ and
\begin{align}
\oh \frac{\ud (\Delta L_1 + \Delta L_2)^2}{\ud t} &=  (\Delta L_1 + \Delta L_2) \frac{\ud (\Delta L_1 + \Delta L_2)}{\ud t} \\
\oh \frac{\ud (\Delta L_1 - \Delta L_2)^2}{\ud t} &=  (\Delta L_1 - \Delta L_2) \frac{\ud (\Delta L_1 - \Delta L_2)}{\ud t}.
\end{align}

We use \cref{eq:active1_interm2,eq:active2_interm2} to form a Lyapunov function for the dynamics. Let
\begin{equation}
\phi(\Delta L_1,\Delta L_2) = \oh \left[\left(\Delta L_1+\Delta L_2\right)^2+\frac{\pi_3}{\pi_1}\left(\Delta L_1-\Delta L_2\right)^2\right].
\end{equation}
\begin{claim*}
$\phi$ is a Lyapunov function on $\Omega$.
\end{claim*}
\noindent We first comment on some implications of the above claim before proceeding with its proof. It follows from the claim that there is a unique steady-state $L_\text{ss}$ at which $(L_1,L_2)=(L_\text{ss},L_\text{ss}) \in \Omega$ and that no limit cycles exist in $\Omega$.

\begin{proof}
Since it is the sum of squared terms with positive factors, $\phi \ge 0$. It remains to show that $\phi' \le 0$ with equality achieved only at the unique steady-state $L_1=L_2=L_\text{ss}$. Because the second term of $\phi$ has a non-positive derivative (by \eqref{eq:active2_interm2}), a sufficient condition for $\phi' \le 0$ is that the right-hand side of \eqref{eq:active1_interm2} is negative away from the steady-state, i.e.~
\begin{equation}
2 \left(J-J_\text{ss}\right)\left(1-2L_\text{ss}-\pi_1 L_\text{ss}\right)\left(\Delta L_1+\Delta L_2\right)-\left(2+\pi_1\right)J\left(\Delta L_1+\Delta L_2\right)^2 < 0,
\label{eq:lyap_claim}
\end{equation}
except when $\Delta L_1 =\Delta L_2=0$ in which case it is zero. As we show next, in some cases \eqref{eq:lyap_claim} holds and the claim immediately follows. Later on, we show how, if \eqref{eq:lyap_claim} does not hold, we may use the second term of $\phi$ to obtain the desired result $\phi' \le 0$.

It follows from the steady-state condition \eqref{eq:active_ss} that
\begin{equation}
1-2L_\text{ss}-\pi_1L_\text{ss} = \frac{\pi_0}{J_\text{ss}}.
\end{equation}
Therefore, the first term on the left-hand side of \eqref{eq:lyap_claim} may be rewritten as
\begin{align}
2 \left(J-J_\text{ss}\right)\left(1-2L_\text{ss}-\pi_1 L_\text{ss}\right)\left(\Delta L_1+\Delta L_2\right)&= \frac{2\pi_0}{J_\text{ss}}\left(J-J_\text{ss}\right)\left(\Delta L_1+\Delta L_2\right) \notag \\
&= 2\pi_0J_\text{ss}^{-1}\left(\frac{1}{J^{-1}}-\frac{1}{J_\text{ss}^{-1}}\right)\left(\Delta L_1+\Delta L_2\right) \notag \\
&= 2\pi_0J_\text{ss}^{-1}\left(\frac{J_\text{ss}^{-1}-J^{-1}}{J_\text{ss}^{-1}J^{-1}}\right)\left(\Delta L_1+\Delta L_2\right) \notag\\
&= 2\pi_0J\left(J_\text{ss}^{-1}-J^{-1}\right)\left(\Delta L_1+\Delta L_2\right).
\label{eq:lyap_claim_int0}
\end{align}
To simplify the right-hand side of the above expression, we use the definition of the flux \eqref{eq:J_nondim} to get
\begin{align}
J_\text{ss}^{-1}-J^{-1} &= 1+2\pi_2L_\text{ss}+2\pi_3L_\text{ss}^2-\left(1+\pi_2\left(L_1+L_2\right)
+\pi_3\left(L_1^2+L_2^2\right)\right) \notag \\
&=-\pi_2\left(\Delta L_1+\Delta L_2\right)-\pi_3\left(L_\text{ss}^2-L_1^2+L_\text{ss}^2- L_2^2\right). \label{eq:lyap_claim_int1}
\end{align}
Note that
\begin{align}
L_1^2-L_\text{ss}^2+L_2^2-L_\text{ss}^2  &= (L_1+L_\text{ss})\Delta L_1+(L_2+L_\text{ss})\Delta L_2 \notag \\
&= L_1\Delta L_1+L_2\Delta L_2+L_\text{ss}(\Delta L_1+\Delta L_2), \label{eq:lyap_claim_int2}
\end{align}
and moreover
\begin{align}
L_1\Delta L_1+L_2\Delta L_2 &= \frac{L_1+L_2}{2}\left(\Delta L_1+\Delta L_2\right)+\left(L_1-\frac{L_1+L_2}{2}\right)\Delta L_1+\left(L_2-\frac{L_1+L_2}{2}\right)\Delta L_2 \notag \\
&= \frac{L_1+L_2}{2}\left(\Delta L_1+\Delta L_2\right)+\frac{L_1-L_2}{2}\Delta L_1-\frac{L_1-L_2}{2}\Delta L_2 \notag \\
&= \frac{L_1+L_2}{2}\left(\Delta L_1+\Delta L_2\right)+\oh\left(\Delta L_1-\Delta L_2\right)^2 . \label{eq:lyap_claim_int3}
\end{align}
Combining \cref{eq:lyap_claim_int0,eq:lyap_claim_int1,eq:lyap_claim_int2,eq:lyap_claim_int3}, for the left-hand side of \eqref{eq:lyap_claim} we have
\begin{align}
2\pi_0J\left(J_\text{ss}^{-1}-J^{-1}\right)\left(\Delta L_1+\Delta L_2\right) &= -2\pi_0 J\left(\pi_2\left(\Delta L_1+\Delta L_2\right)^2+\pi_3 L_\text{ss}\left(\Delta L_1+\Delta L_2\right)^2 \right. \notag \\
&\left.+\frac{\pi_3}{2}\left(L_1+L_2\right)  \left(\Delta L_1+\Delta L_2\right)^2+\frac{\pi_3}{2} \left(\Delta L_1-\Delta L_2\right)^2\left(\Delta L_1+\Delta L_2\right)\right). \label{eq:lyap_claim_int4}
\end{align} 
To show the above expression is negative, we consider two cases. In the case $\Delta L_1+\Delta L_2 > 0$, the term in parentheses on the right-hand side of \eqref{eq:lyap_claim_int4} is a sum of squared terms multiplied by positive numbers. Therefore the right-hand side of \eqref{eq:lyap_claim_int4} is negative.

In the case $\Delta L_1+\Delta L_2 < 0$, we may bound the positive term by the negative terms in the right-hand side of \eqref{eq:lyap_claim_int4} to show that the overall expression is negative. In this case, $(L_1,L_2)$ must reside in a triangular subset $\Omega^t \subset \Omega$ (the striped region in Fig.~\ref{fig:proof_schema}).
\begin{figure}[]
\centering
	 \includegraphics[width=0.4\textwidth]{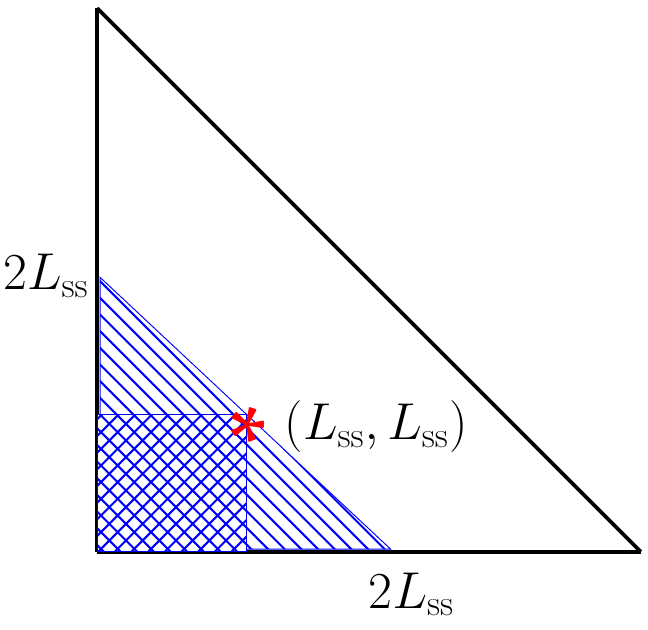}
         \caption{Schematic of the triangular region $\Omega^t$ used in the proof. $\Omega^t$ is illustrated by the diagonal pattern; it is defined as the region in which $\Delta L_1 + \Delta L_2 <0$. The subset of $\Omega^t$ in which both $\Delta L_1<0$ and $\Delta L_2<0$ is illustrated by the crosshatched pattern.}
         \label{fig:proof_schema}
\end{figure}
On this subset, if both $\Delta L_1$ and $\Delta L_2$ are negative, then
\begin{equation}
\abs{\Delta L_1-\Delta L_2} < \abs{\Delta L_1+\Delta L_2},
\end{equation}
and moreover $\abs{\Delta L_1-\Delta L_2} < \abs{\Delta L_1}+\abs{\Delta L_2} <2L_\text{ss}$,
so that
\begin{equation}
\frac{\pi_3}{2} \abs{\Delta L_1-\Delta L_2}^2\abs{\Delta L_1+\Delta L_2} \le 
\pi_3 L_\text{ss}\abs{\Delta L_1+\Delta L_2}^2,
\end{equation}
and since this term already appears in \eqref{eq:lyap_claim_int4} with a minus sign it follows that the right-hand side of \eqref{eq:lyap_claim_int4} is negative.

On the other hand, if $\Delta L_1$ and $\Delta L_2$ have opposite signs, then we must use a different bound. From the geometry of the relevant triangular subset $\Omega^t \subset \Omega$, we have as before $\abs{\Delta L_1+\Delta L_2} \le 2L_\text{ss}$ so that
\begin{equation}
\frac{\pi_3}{2} \abs{\Delta L_1-\Delta L_2}^2\abs{\Delta L_1+\Delta L_2} \le 
\pi_3 L_\text{ss}\abs{\Delta L_1-\Delta L_2}^2.
\end{equation}
From the condition that the flagella initially grow if they start at zero length, we have $\pi_0 < 1$, and from the steady-state equation \eqref{eq:active_ss} it follows that $L_\text{ss} \le 1/2$. Therefore
\begin{align}
2\pi_0 J \pi_3 L_\text{ss}\abs{\Delta L_1-\Delta L_2}^2 < \pi_3 J \abs{\Delta L_1-\Delta L_2}^2.
\end{align}
It follows that $(\pi_3/2) \abs{\Delta L_1-\Delta L_2}^2\abs{\Delta L_1+\Delta L_2} $ is bounded above by $\pi_3/\pi_1\left(\pi_1J\left(\Delta L_1-\Delta L_2\right)^2\right)$, where $\pi_1J\left(\Delta L_1-\Delta L_2\right)^2$ is the right-hand side of \eqref{eq:active2_interm2}. Therefore,
\begin{equation}
\frac{\ud \phi}{\ud t} =\oh \left(\frac{\ud (\Delta L_1 + \Delta L_2)^2}{\ud t}+\frac{\pi_3}{\pi_1}\frac{\ud (\Delta L_1 - \Delta L_2)^2}{\ud t}\right) < 0.
\end{equation}
\end{proof}

Next, we consider whether $\Omega$ is a trapping region closed to the dynamics. That is, we must consider whether any trajectories originating in $\Omega$ may eventually leave $\Omega$. First, we show that no trajectories cross the $L_1+L_2=1$ boundary. To show that all trajectories move inward from this boundary, we calculate
\begin{align}
 \begin{pmatrix}
  L_1 \\ L_2
 \end{pmatrix}'\Big|_{L_1+L_2=1} \cdot
 \begin{pmatrix}
  -1 \\ -1
 \end{pmatrix}
 &= 
  -(J(1-1)-\pi_0-\pi_1JL_1)-(J(1-1)-\pi_0-\pi_1JL_2) \notag \\
  &=2\pi_0+\pi_1J(L_1+L_2) > 0.
\end{align}
This proves that trajectories flow inward cannot leave $\Omega$ across the $L_1+L_2=1$ boundary. See Figure \ref{fig:lyap_domain} for an illustration of $\Omega$ as well as the slope field and contours of $\phi(\Delta L_1,\Delta L_2)$.
\begin{figure}[]
         \centering
	 \includegraphics[width=0.9\textwidth]{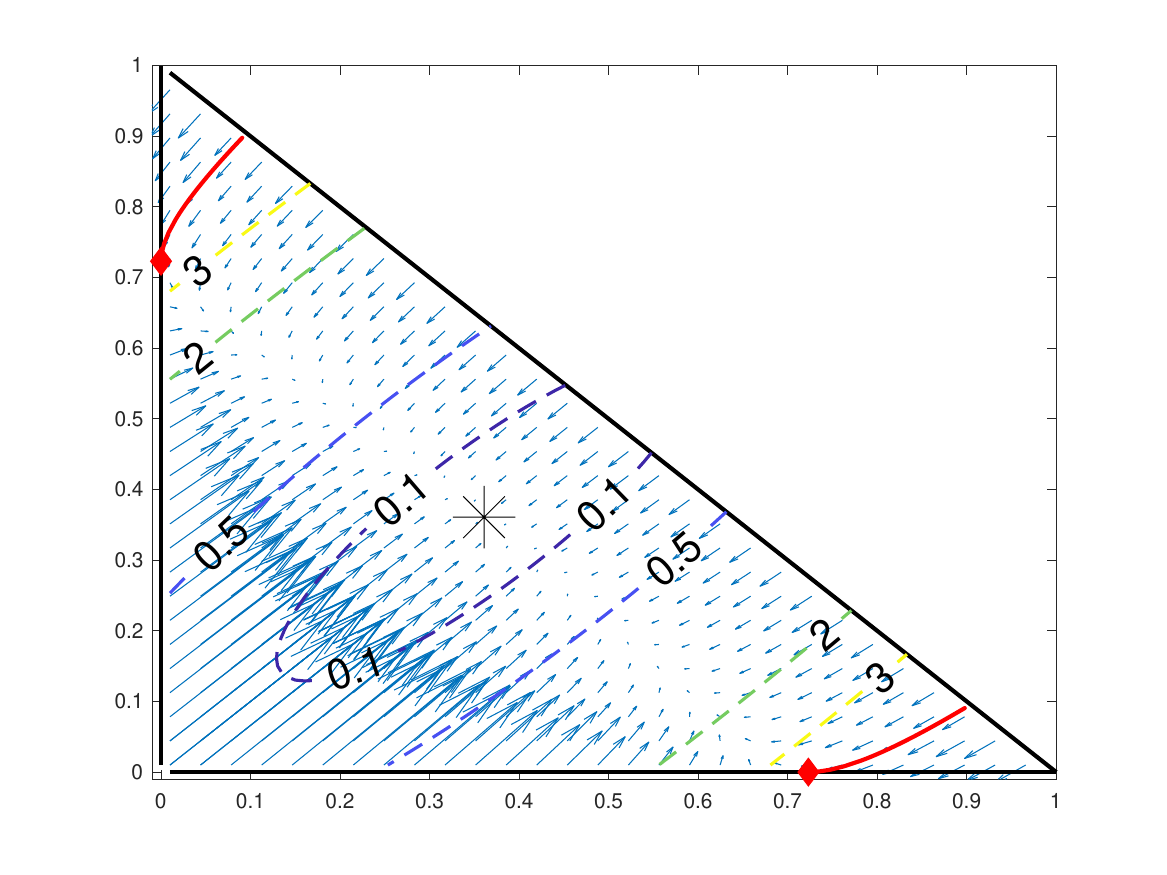}
         \caption{Dynamics of the active disassembly model. The slope field is plotted and contours $\phi(\Delta L_1,\Delta L_2) = C$ of the Lyapunov function are shown for several values of $C$. The black asterisk denotes the unique steady state, the black curves represent boundaries of $\Omega$, and the red diamonds represent the points $(L_\text{crit},0)$ and $(0,L_\text{crit})$ to which trajectories contract from the corner regions. The parameters used are $\pi_0=0.1$, $\pi_1=0.15$, $\pi_2=1$, and $\pi_3=2$.}
         \label{fig:lyap_domain}
\end{figure}

In addition, for $\Omega$ to be a trapping region, no trajectories should escape the first quadrant. In terms of $\Omega$, this means that no trajectories cross the $L_1$ axis on the interval $0 < L_1 \le 1$, and no trajectories cross the $L_2$ axis on the interval $0 < L_2 \le 1$. However, the model \cref{eq:active1_nondim,eq:active2_nondim} does not satisfy this requirement. To the contrary,
\begin{align}
\begin{pmatrix}
  L_1 \\ L_2
 \end{pmatrix}'\Big|_{0 < L_1 \le 1, L_2 = 0} \cdot
 \begin{pmatrix}
  0 \\ 1
 \end{pmatrix}
 &= 
J (1-L_1)-\pi_0,
\end{align}
and upon solving the quadratic equation resulting from $ J (T-L_1)-\pi_0= 0$ one finds that there is a unique positive root
\begin{equation}
L_\text{crit} = \frac{1+\pi_0 \pi_2}{2\pi_0 \pi_3}\left( \sqrt{1+\frac{4\left(1-\pi_0\right)\pi_0 \pi_3}{\left(1+\pi_0 \pi_2\right)^2}}-1\right).
\end{equation}
Therefore, to prevent trajectories that cross the $L_1$-axis resulting in negative lengths, we construct a hybrid system by setting $L_2'=0$ when $L_1 > L_\text{crit}$ on $L_2=0$. This results in trajectories that flow down to the $L_1$-axis, move horizontally according to this modified dynamics with $L_2'=0$ until they reach $(L_1,L_2)=(L_\text{crit},0)$, at which point they rejoin the full dynamics and move toward steady state. The dynamics are modified symmetrically to prevent trajectories from crossing the $L_2$-axis as well.

To characterize the trapping region of the full dynamics, we therefore define the trajectory $\Gamma_1(t)$ going backward in time from the point $(L_\text{crit},0)$, as defined by
\begin{align}
\Gamma_1(0) = (L_\text{crit},0),
\end{align}
up until it hits the diagonal $L_1+L_2 = 1$. We define another trajectory $\Gamma_2(t)$ similarly as the one going backward in time from the point $(0,L_\text{crit})$. The points from which the trajectories $\Gamma_1(t)$ and $\Gamma_2(t)$ originate are illustrated by red diamonds and the trajectories are shown as red curves in Fig.~\ref{fig:lyap_domain}.

As a consequence of the modified dynamics, trajectories originating in the corner regions of phase space bounded by $\Gamma_1(t)$ and $\Gamma_2(t)$ consist of three distinct phases: first, there is flow down to the $L_1$ or $L_2$-axis, followed next by horizontal or vertical motion to the point $(L_\text{crit},0)$ or $(0,L_\text{crit})$, respectively, and finally leading to recovery by the full dynamics back to steady state. One interesting consequence of this behavior is its interpretation in the context of severing experiments. We find that the typical severing protocol, in which the two flagella first reach steady-state before severing takes place, does not lead to states within the corner region (purple trajectory in Fig.~\ref{fig:lyap_sever}). 
\begin{figure}[]
         \vspace{-.5in}
         \centering
	 \includegraphics[width=.8\textwidth]{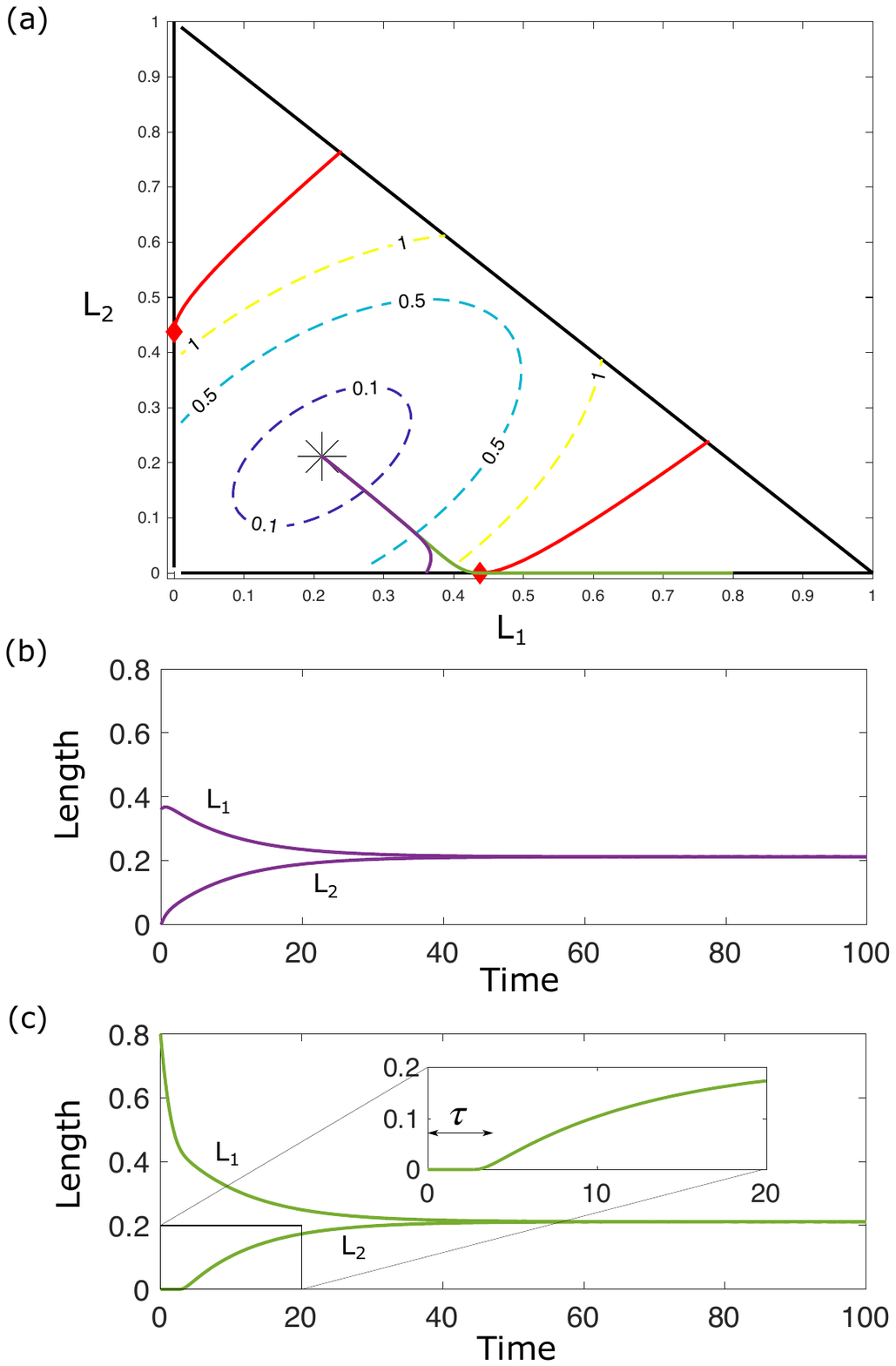}
	 \caption{Upon severing, the pool size instantaneously decreases by the steady-state length $L_\text{ss}$. This leads to a rescaling of parameters. (a) A severing experiment is illustrated by plotting two trajectories in the $L_1$-$L_2$ phase plane in green and purple, (b) and (c) Length vs. time corresponding respectively to the green and purple trajectories in (a). To allow for direct comparison to Fig.~\ref{fig:lyap_domain}, the nondimensional parameters are obtained using the initial tubulin pool size prior to severing. The green trajectory is an example in which the flagella are initially equal at a value greater than the steady-state length. There is an initial delay $\tau$ as the longer flagellum shortens, during which the severed flagellum remains at zero length (see inset in (c)).}
\label{fig:lyap_sever}
\end{figure}

Fig.~\ref{fig:lyap_sever} shows a representative simulation of a severing experiment. The code used to generate this simulation and the corresponding figures is included in our open source code repository located at \url{https://github.com/thomasgfai/StabilityFlagellarLengths}.

If the system is initially out of steady-state, however, severing may result in a state outside the trapping region. This yields two phases of flow after severing: first, a horizontal flow to the accumulation point $L_\text{crit}$, and second a flow by the full dynamics to reach the global steady-state (green trajectory in Fig.~\ref{fig:lyap_sever}). As shown in the inset of Fig.~\ref{fig:lyap_sever}(c), this first phase of horizontal flow leads to an apparent delay in the response of the cut flagellum.

{\color{black} As mentioned above, the corner regions that lead to negative lengths in the original active disassembly model are not encountered in typical biological scenarios. One may nevertheless ask} how a model derived from physical principles could lead to negative lengths. In our case, this behavior arises because of the basal disassembly rate $d_0$ that appears in the model, which implies a nonzero rate of disassembly even when the flagellum is very small. In the corner regions of the phase space, this basal term is greater than the rate of assembly because the tubulin pool is nearly depleted. This leads to a negative growth rate at zero length. Note that this issue does not arise in agent-based models in which each tubulin monomer is treated explicitly, since in that case disassembly only occurs when the flagellum contains at least one monomer. Within the ODE model, the modified dynamics ensure that disassembly only occurs when there are tubulin monomers remaining on the flagellum. {\color{black} Alternatively, one could rectify this issue by defining $d_0=d_0(L)$ to be a step function of length that turns off once a flagellum is sufficiently small.}

\subsection{Nonexistence of Oscillations: Alternate Proof}
Recall that we have an invariant set on the domain $\Omega$ by construction. In this section, we use an alternative approach to exclude oscillations using standard methods in dynamical systems. The proof will primarily focus on the region of $\Omega$ containing the fixed point $(L_\text{ss},L_\text{ss})$ that is bounded by the lines $L_2=L_1\pm \ve$:
\begin{equation*}
\Omega_\ve := \{(L_1,L_2) \,| \, 0 \leq L_1,\, 0\leq L_2, \,L_1 + L_2\leq T,\,L_2 \leq L_1 + \ve,\,L_1 - \ve \leq L_2 \},
\end{equation*}
where $0< \ve \leq L_\text{crit}$. An example of this domain is shown in Figure \ref{fig:alt} for $\ve=L_\text{crit}$ (checkerboard region). The boundary of $\Omega_\ve$ for $\ve=0.2$ is shown using dotted lines, labeled $L_2 = L_1\pm \ve$. The smaller the value of $\ve$, the thinner the domain $\Omega_\ve$ will be. We will first show that $\Omega_\ve$ is a trapping region (invariant set) for each $\ve$ arbitrarily small, by showing that the flows across the sections defined by the lines $L_2=L_1\pm \ve$ are always inward towards the fixed point $L_{ss}$. This condition is sufficient because we have already shown that flow along the relevant boundaries of $\Omega$ point inwards, specifically, $(L_1,0)$ and $(0,L_2)$ for $L_1,L_2 \in [0,L_\text{crit}]$, and $L_1+L_2 = 1$. To prove the non-existence of oscillations, we assume it exists and show that a portion of the oscillation must escape $\Omega_\ve$ for some $\ve$ and can not return, which is a contradiction.
\begin{figure}
 	\vspace{-.5in}
 	\centering
 	\includegraphics[width=.7\textwidth]{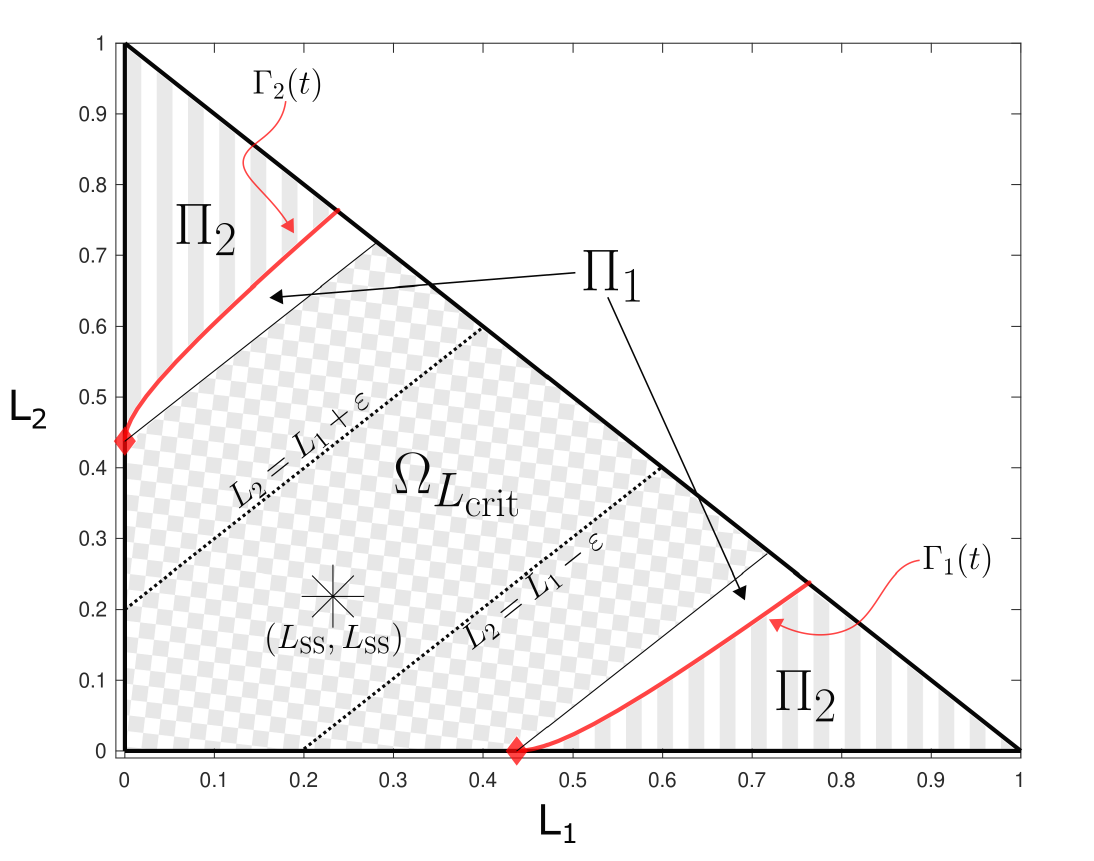}
 	\caption{Domains of the alternative proof. $\Omega_{L_\text{crit}}$ (checkerboard region) denotes the region $\Omega_\ve$ for $\ve=L_\text{crit}$. Dotted lines show an example boundary of $\Omega_\ve$ for $\ve=0.2$. $\Pi_1$ (white region) denotes the region between $\Omega_{L_\text{crit}}$ and the trajectories $\Gamma_1(t)$ and $\Gamma_2(t)$ (red curves). $\Pi_2$ (striped region) denotes the region outside of $\Omega_{L_\text{crit}}$ and $\pi_1$ but still inside the overall triangular domain $\Omega$. The parameters are the same as in Figure \ref{fig:lyap_sever}.}
 	\label{fig:alt}
 \end{figure}

\begin{claim*}
$\Omega_\ve$ is an invariant set.
\end{claim*}

In the ensuing proof, we show that the flows across the sections defined by the lines $L_2=L_1\pm \ve$ are always inward towards the fixed point $L_{ss}$. Combined with the inward flow along the boundaries of $\Omega$, this inward flow across the lines $L_2=L_1\pm \ve$ is sufficient to prove the claim.
\begin{proof}
To show that the vector field points inwards for a given $\Omega_\ve$, we show that $dL_1/dt > dL_2/dt$ along $L_2=L_1+\ve$ and $\ud L_1/ \ud t < \ud L_2/ \ud t$ along $L_2 = L_1 - \ve$. First, we plug in $L_2=L_1+\ve$ into the right hand side of \cref{eq:active1_nondim,eq:active2_nondim} and take the difference $\ud L_1/ \ud t-\ud L_2/ \ud t$:
\begin{align*}
\frac{\ud L_1}{\ud t} - \frac{\ud L_2}{\ud t} &= J(1-L_1-L_2)-\pi_0-\pi_1 J L_1-\left[J(1-L_1-L_2)-\pi_0-\pi_1 JL_2\right]\\
&= -\pi_1JL_1 + \pi_1 J(L_1+\ve)\\
&=\pi_1 J \ve > 0.
\end{align*}
Therefore, $\ud L_1/ \ud t>\ud L_2/\ud t$ independent of the choice of $\ve$. Inward flow along the lower boundary follows using the same argument. Thus, for each $\ve$, the set ${\Omega}_\ve$ is invariant.

We now proceed with a proof by contradiction. Suppose that there exists a periodic solution. Every oscillating solution must surround at least one fixed point (see index theory, \cite{guckenheimer1983local}), and the only fixed point of this system exists within ${\Omega}_\ve$. Therefore, the oscillating solution surrounds the fixed point $(L_\text{ss},L_\text{ss})$, and there exists some $\ve^*$ such that parts of the periodic solution lie outside of ${\Omega}_{\ve^*}$. (If an oscillating solution is fully contained within ${\Omega}_\ve$ for some $\ve$, to obtain $\ve^*$ we may simply reduce $\ve$ until parts of the periodic solution lie outside of ${\Omega}_{\ve^*}$.) It follows that solutions must exit ${\Omega}_{\ve^*}$, but this contradicts ${\Omega}_{\ve^*}$ being an invariant set.
\end{proof}

\subsection{Global Asymptotic Stability: Alternate Proof}
\begin{claim*}
The fixed point $(L_\text{ss},L_\text{ss})$ is globally asymptotically stable in $\Omega$.
\end{claim*}
\begin{proof}
We will consider three parts of the domain. First, $\Omega_{L_\text{crit}}:=\Omega_\ve$ for $\ve = L_\text{crit}$ (Figure \ref{fig:alt}, checkerboard region), second, the region between $\Omega_{L_\text{crit}}$ and $\Gamma_1$ which we call $\Pi_1$  (Figure \ref{fig:alt}, white region within $\Omega$), and third, the region $\Pi_2 := \Omega \setminus (\Omega_{L_\text{crit}} \cup \Pi_1)$  (Figure \ref{fig:alt}, striped region). The sets $\Omega_{L_\text{crit}}$, $\Pi_1$, and $\Pi_2$ are taken to be closed.

We begin our proof with $\Pi_1$ and $\Pi_2$ because they are the most straightforward. In $\Pi_1$, all solutions exit through one of the lines $L_2=L_1\pm L_\text{crit}$. This observation follows because there are no fixed points in $\Pi_1$ and flows cannot cross the solution $\Gamma_1$. In $\Pi_2$, we force all flows to exit through $(L_\text{crit},0)$ or $(0,L_\text{crit})$ by definition. Therefore, all flows originating in $\Pi_1 \cup \Pi_2$ will enter the trapping region $\Omega_{L_\text{crit}}$.

It remains to show asymptotic stability in $\Omega_{L_\text{crit}}$. Note that our system is continuously differentiable on all of $\Omega_{L_\text{crit}}$ and therefore existence and uniqueness of solutions holds for all time. Before proceeding with the proof of asymptotic stability, we recall two basic definitions.
\begin{definition}
	The \textbf{positive semiorbit} is defined as the set
	\begin{equation*}
	\Gamma(\bm{p}) = \{ \bm{x} \in \mathbb{R}^2\, : \, \bm{x} = \phi(t,\bm{p}), \, \text{for some} \,\,\, t \in (0,\infty)\}.
	\end{equation*}
\end{definition}

\begin{definition}
	The \textbf{$\omega$-limit set} of a point $\bm{x} \in M$, $\omega(\bm{x})$, is the set of those points $\bm{y} \in M$ for which there exists a sequence $t_n \rightarrow \infty$ with $\phi(t_n,\bm{x}) \rightarrow \bm{y}$, where $\phi$ is the flow of an autonomous ODE $\dot{\bm{x}} = \bm{F}(\bm{x})$, and $M$ is an open subset of $\mathbb{R}^2$ \cite{teschl2012ordinary}. The omega-limit set $\omega(\Gamma)$ for a trajectory $\Gamma$ is defined analogously.
\end{definition}

In the case that all solutions are bounded, only certain types of $\omega$-limit sets are possible. This is a consequence of the following theorem (Theorem 7.4.2 in \cite{lebovitz_2019}):
\begin{theorem}\label{theorem:bounded}
	Let $\Gamma$ be a positive semiorbit contained in a compact subset $K$ of $\mathbb{R}^2$ and suppose that $K$ contains only a finite number of equilibrium points. Then
	\begin{enumerate}
	        \item $\omega(\Gamma)$ consists of a single point $\bm{p}$ which is an equilibrium point of $\bm{F}$, and for an initial condition $\bm{x}(0)\in\gamma$, $\bm{x}(t)\rightarrow \bm{p}$ as $t \rightarrow \infty$, or
		\item $\omega(\Gamma)$ is a periodic orbit, or
		\item $\omega(\Gamma)$ consists of a finite set of equilibrium points together with their connecting orbits. Each such orbit approaches an equilibrium point on $t \rightarrow \infty$ and as $t \rightarrow -\infty$.
	\end{enumerate}
\end{theorem}
This theorem, which follows from the Poincar\'{e}-Bendixson theorem, tells us that if a solution is trapped in a compact domain for all time, then a single equilibrium must be either a unique asymptotically stable point or a saddle-like point with at least one pair of connected unstable and stable manifolds. (The saddle-like observation follows by the third part of the theorem; a single equilibrium can also be stable in backwards time.)

Because we have eliminated periodic orbits in $\Omega_{L_\text{crit}}$ and shown that our stable fixed point $(L_\text{ss},L_\text{ss})$ is locally asymptotically stable and therefore cannot be saddle-like, we must have an asymptotically stable equilibrium point in $\Omega_{L_\text{crit}}$. Finally, because all solutions in $\Omega\setminus\Omega_{L_\text{crit}}$ eventually reach the trapping region $\Omega_{L_\text{crit}}$, $(L_\text{ss},L_\text{ss})$ is the globally asymptotically stable equilibrium point in $\Omega$.
\end{proof}

\subsection{Generalization to organisms with $N >2$ flagella}
Next, we prove global asymptotic stability in the case of arbitrary flagellar number by constructing a Lyapunov function, following similar logic to that used in the $N=2$ case. In the following, we once again drop the tildes for convenience. Define $\Delta L_i = L_i-L_\text{ss}$ as before for $i=1,\dots,N$ and let
\begin{equation}
\phi(\Delta L_1,\dots,\Delta L_N) = \oh \left[\left(\sum_{i=1}^N \Delta L_i \right)^2+\frac{2\pi_3}{\pi_1}\sum_{i\neq j}\left(\Delta L_i-\Delta L_j\right)^2\right].
\label{eq:claimN}
\end{equation}
\begin{claim*}
$\phi$ is a Lyapunov function on $\Omega$.
\end{claim*}

\begin{proof}
We first note that setting the derivative to zero in \eqref{eq:activeN_nondim} yields the identify
\begin{align}
\label{eq:active_ssN}
0=J_\text{ss}\left(1-NL_\text{ss}\right)-\pi_0-\pi_1J_\text{ss}L_\text{ss}.
\end{align}
where
\begin{align}
J_\text{ss}=\frac{1}{1+N\pi_2 L_\text{ss}+N\pi_3 L_\text{ss}^2}.
\end{align}
Following the logic of the proof in the case $N=2$ results in the equations
\begin{align}
\oh \frac{\ud}{\ud t}\left(\sum_{i=1}^N \Delta L_i\right)^2 &= N \left(J-J_\text{ss}\right)\left(1-NL_\text{ss}-\pi_1 L_\text{ss}\right)\sum_{i=1}^N \Delta L_i-\left(N+\pi_1\right)J\left(\sum_{i=1}^N \Delta L_i\right)^2  \label{eq:active1N}\\
\oh \frac{\ud}{\ud t}\left( \Delta L_i-\Delta L_j\right)^2 &= -\pi_1J\left(\Delta L_i-\Delta L_j\right)^2,\quad \text{for }i\neq j. \label{eq:active2N}
\end{align}
From the steady-state identify \eqref{eq:active_ssN}, we have
\begin{equation}
1-NL_\text{ss}-\pi_1 L_\text{ss} = \frac{\pi_0}{J_\text{ss}},
\end{equation}
so that the first term on the right-hand side of \eqref{eq:active1N} may be rewritten as
\begin{align}
N \left(J-J_\text{ss}\right)\left(1-NL_\text{ss}-\pi_1 L_\text{ss}\right)\sum_{i=1}^N \Delta L_i
&=  \frac{N\pi_0}{J_\text{ss}}\left(J-J_\text{ss}\right)\sum_{i=1}^N \Delta L_i \notag \\
&= N\pi_0 J\left(J_\text{ss}^{-1}-J^{-1}\right)\sum_{i=1}^N \Delta L_i.
\label{eq:lyap_claim_int1N}
\end{align}
The expression above may be simplified further by using the definition of the flux to obtain
\begin{align}
J_\text{ss}^{-1}-J^{-1} = -\pi_2 \sum_{i=1}^N \Delta L_i -\pi_3 \sum_{i=1}^N \left(L_i^2-L_\text{ss}^2\right),
\label{eq:lyap_claim_int2N}
\end{align}
and noting that
\begin{align}
\sum_{i=1}^N \left(L_i^2-L_\text{ss}^2\right) &= \sum_{i=1}^N\left(L_i+L_\text{ss}\right)\Delta L_i \notag \\
&=  \sum_{i=1}^N L_i \Delta L_i+L_\text{ss} \sum_{i=1}^N \Delta L_i.
\label{eq:lyap_claim_int3N}
\end{align}
Moreover, we may rewrite the first term on the right-hand side above by adding and subtracting $(\frac{1}{N}\sum_{i=1}^N L_i)\sum_{i=1}^N \Delta L_i$ to obtain
\begin{align}
\sum_{i=1}^N L_i \Delta L_i &= \left(\frac{1}{N}\sum_{i=1}^N L_i\right)\sum_{i=1}^N \Delta L_i+\sum_{i=1}^N \left(L_i-\frac{1}{N}\sum_{i=1}^N L_i\right)\Delta L_i \notag \\
&= \left(\frac{1}{N}\sum_{i=1}^N L_i\right)\sum_{i=1}^N \Delta L_i+\frac{1}{N}\sum_{i=1}^N\left(\sum_{j\neq i} \Delta L_i-\Delta L_j\right)\Delta L_i.
\label{eq:lyap_claim_int4N}
\end{align}
Combining \cref{eq:lyap_claim_int1N,eq:lyap_claim_int2N,eq:lyap_claim_int3N,eq:lyap_claim_int4N} yields
\begin{align}
N\pi_0 J\left(J_\text{ss}^{-1}-J^{-1}\right)\sum_{i=1}^N \Delta L_i =& -N\pi_0 J\left(
\pi_2 \left(\sum_{i=1}^N \Delta L_i\right)^2+\pi_3 L_\text{ss} \left(\sum_{i=1}^N \Delta L_i\right)^2
\right.
\notag \\
&\left.+\frac{\pi_3}{N}\left(\sum_{i=1}^N L_i \right)\left(\sum_{i=1}^N\Delta L_i\right)^2
+\frac{\pi_3}{N}\left(\sum_{i=1}^N \Delta L_i\right)\sum_{i\neq j}\left(\Delta L_i-\Delta L_j\right)^2
\right).
\label{eq:lyap_claim_int5N}
\end{align}
The first three terms in parentheses on the right-hand side of \eqref{eq:lyap_claim_int5N} are clearly positive. The final term may be bounded according to
\begin{equation}
\left|\frac{\pi_3}{N}\left(\sum_{i=1}^N \Delta L_i\right)\sum_{i\neq j}\left(\Delta L_i-\Delta L_j\right)^2\right|
\le \frac{\pi_3}{N}\left|\sum_{i=1}^N \Delta L_i \right| \left| \sum_{i\neq j}\left(\Delta L_i-\Delta L_j\right)^2 \right|
\le \frac{2 \pi_3}{N}\sum_{i\neq j} \left(\Delta L_i-\Delta L_j\right)^2,
\label{eq:lyap_claim_int6N}
\end{equation}
where the final inequality follows from $\left|\sum_{i=1}^N \Delta L_i \right| \le \left|\sum_{i=1}^N L_i \right|+ \left|\sum_{i=1}^N L_\text{ss} \right| \le 2$, by \eqref{eq:active_ssN} and the definition of $\Omega$, which asserts that the sum of lengths must not exceed the total pool size.

Therefore, substituting \cref{eq:lyap_claim_int5N,eq:lyap_claim_int6N} into \cref{eq:active1N,eq:active2N} results in
\begin{align}
\frac{\ud\phi}{\ud t} &\le 2 \pi_3\sum_{i\neq j} \left(\Delta L_i-\Delta L_j\right)^2-\frac{2\pi_3}{\pi_1}\sum_{i\neq j}\left(\pi_1 J \left(\Delta L_i-\Delta L_j\right)^2\right) \notag \\
&\le2 \pi_3\sum_{i\neq j} \left(\Delta L_i-\Delta L_j\right)^2-2 \pi_3\sum_{i\neq j} \left(\Delta L_i-\Delta L_j\right)^2\le0,
\end{align}
with equality obtained at the unique steady-state, i.e.~$\phi$ satisfies the requirement of a Lyapunov function.
\end{proof}
Note that the Lyapunov function we construct for general $N$ does not reduce to the Lyapunov function constructed previously in the case $N=2$. There is an additional factor of two in the second term of $\phi$. However, there is no inconsistency; the Lyapunov function is not unique. Still, the Lyapunov function constructed previously provides a tighter bound in the case $N=2$, and we conjecture that the factor of two in the second term of the Lyapunov function in \eqref{eq:claimN} is not essential.
 
\section{Discussion}
{\color{black} We have modified the active disassembly model originally proposed in \cite{fai2019length} by constructing a hybrid system to ensure that the flagellar lengths remain positive for all time, regardless of the initial conditions.} Further, we have used the method of Lyapunov functions (as well as the Poincar\'{e}-Bendixson Theorem) to prove the existence of an asymptotically stable steady-state. Although the Poincar\'{e}-Bendixson Theorem is limited to systems of two equations, the Lyapunov approach is shown to generalize to $N>2$.

The Lyapunov function provides significant constraints on the possible dynamics, e.g. there can be no periodic cycles or finite time blow-up and solutions flow inextricably down the gradients of the Lyapunov function toward the unique steady-state. This highly-constrained solution space is a testable prediction of the active disassembly model.
 


Finally, we discuss aspects of length control in the case of a large number of organelles. Indeed, this is a biologically-relevant case given that single-cell organisms such as \emph{Paramecium} cells swim using on the order of five thousand cilia \cite{funfak2015paramecium}. Depending on which parameters are fixed in the limit $N\to \infty$, our dimensional analysis of the active disassembly model yields different limiting behaviors. A remarkable consequence of the model is that it provides a universal length-control mechanism regardless of organelle number, and an important next step would be to establish the limits of the model by testing its predictions on different organisms over a range of $N$.

This generalization to arbitrary flagellar number reveals that while the active disassembly model was originally motivated by \emph{Chlamydomonas reinhardtii} with its two flagella, it is universal in the sense that it achieves length control for any number of organelles. {\color{black} We have further explore the implications of this universal mechanism for controlling the sizes of large numbers of organelles by applying the model to study olfactory sensory neurons. The model provides possible explanations for scaling relations between cell size and flagellar length observed in the data.}

{\color{black} An important caveat is that in organisms with more than two flagella, some of the assumptions of our model may need to be modified. For example, \emph{Giardia} has eight flagella organized in four pairs, with each pair having a different steady-state length \cite{mcinally2016eight}. In order to represent length control in \emph{Giardia} using the active disassembly model, the same overall modeling framework could be used with suitable modifications. For instance, each flagellar pair could be taken to have different parameters, such as the injection rate at the flagellar base, for example, or the tubulin and/or motor pools may be taken to be separate containing different amounts of protein for each of the four pairs.}

\section{Acknowledgments}
We thank Ariel Amir and Te Cao for providing feedback on an early version of this manuscript. We further acknowledge useful discussions with Prathitha Kar, Lishibanya Mohapatra, and Jane Kondev, and Rosemary Challis for providing data on cilia lengths in olfactory sensory neurons. We acknowledge support under National Science Foundation grant DMS-1913093 (TGF) and National Institute of Health grant T32 NS007292 (YP).

\begin{appendices}
{\color{black}
\section{Flux derivation}
\label{app:flux}
Consider the lengths of two flagella and their evolution in time, which we denote by $L_1(t)$ and $L_2(t)$.   As stated in Section \ref{sec:model}, we assume a total number of molecular motors $M$ and a total amount of tubulin $T$.

We now derive the dynamical equations for the flagellar lengths, closely following \cite{fai2019length}. Assembly occurs in two steps in our model, each of which follows the law of first-order chemical kinetics. First, tubulin aggregates on an IFT particle containing a single molecular motor. Second, the IFT particles are injected into the flagellum. The first step of protein aggregation on IFT particles results in particles carrying an amount of tubulin proportional to $T_f$. In the second step, i.e.~the injection of IFT particles at the flagellar base, the number $M_f$ of free molecular motors determines the injection rate $J_i$ via:
 \begin{equation}
 J_i= \oh k_\text{on} M_f,
 \end{equation}
  where $k_\text{on}$ is the rate constant of injection and the factor of one half is due to equal probabilities of injection into either flagellum.  
    
 Because the timescale of IFT (order of seconds) is fast compared to the timescale of changes in the flagellar lengths (order of minutes), we assume a quasi-steady state in which there is no accumulation of IFT particles along the flagellum. Therefore, the rate of injection equals the arrival flux of IFT particles to the flagellar tip.
 
 The assembly rate is assumed proportional to this arrival flux times the average amount of tubulin carried by an IFT particle, resulting in
 \begin{equation}
\text{assembly rate} = \gamma J_i T_f = \oh \gamma k_\text{on} M_f T_f,
\label{eq:inj}
 \end{equation}
 where $\gamma$ is a constant of proportionality.
 
 As mentioned above $T_f = T-L_1-L_2$. In order to express $M_f$ in terms of the flagellar lengths, we must incorporate the detailed motion of proteins during IFT. While, as mentioned above, IFT particles including the motors and cargo move at constant speed in the anterograde direction, it has been shown that some IFT proteins such as kinesin motors diffuse in the retrograde direction back to the flagellar base \cite{chien2017dynamics}. We capture these different possibilities by letting $M_{b,i}$ and $M_{d_i}$ represent the numbers of motors undergoing ballistic and diffusive motion, respectively, on the $i$\textsuperscript{th} flagellum. By conservation of motor number, we have
 \begin{equation}
 M_f = M-M_{b,1}-M_{b,2}-M_{d,1}-M_{d,2}
 \label{eq:Mcons}
 \end{equation}
 
We may express the ballistic flux $J_{b,i}$ on the $i$\textsuperscript{th} flagellum as the product of the concentration of motors moving ballistically in the anterograde direction times their speed:
\begin{equation}
J_{b,i}= \frac{M_{b,i}v}{L_i},
\label{eq:Mball}
\end{equation}
where $v$ is the IFT speed in the anterograde direction.
 
On the other hand,  To capture this diffusive motion, we let the concentration of diffusing particles on the $i$\textsuperscript{th} flagellum be $c_{d,i}(x)$ for positions $x\in[0,L_i]$ along the flagellum. If the diffusive flux $J_{d,i}$ is constant, then by Fick's law $J = -D \partial c/\partial x$ so that the concentration will be a linear function of $x$, and in particular
\begin{equation}
c_{d,i}(x) = \frac{M_{d,i}}{L_i}+\frac{J_{d,i}}{D}\left(x-L_i/2\right),
\label{eq:Mconc}
\end{equation}
where we have used the fact that $\int_0^{L_i} c_{d,i} \ud x = M_{d,i}$.

Assuming the flagellar base acts as a sink, we may apply the boundary condition $c_{d,i}(0)=0$ to \eqref{eq:Mconc} and express the flux in terms of the number of diffusing motors:
\begin{equation}
 J_{d,i} = \frac{2DM_{d,i}}{L_i^2},
 \label{eq:Mdiff}
 \end{equation}
 and it follows that
 \begin{equation}
c_{d,i}(x) = \frac{J_{d,i}}{D}x.
\label{eq:Mconc2}
\end{equation}

Since $J_i=J_{b,1}=J_{b,2}=J_{d,1}=J_{d,2}$ by the quasi-steady state assumption, we drop the subscripts and denote the constant flux by $J$. Plugging the expressions from \cref{eq:Mcons,eq:Mball,eq:Mdiff} into \eqref{eq:inj} yields
\begin{equation*}
J = \oh k_\text{on}\left(M-J\frac{L_1+L_2}{2v}+J\frac{L_1^2+L_2^2}{4D}\right),
\end{equation*}
and upon solving for $J$ we obtain
\begin{equation}
J = \frac{k_\text{on}M/2}{1+k_\text{on}\left(L_1+L_2\right)/2v+k_\text{on}\left(L_1^2+L_2^2\right)/4D}.
\label{eq:Jsolve}
\end{equation}
Inserting this expression into \eqref{eq:inj} yields finally
 \begin{equation}
\text{assembly rate} = \gamma J \left(T-L_1-L_2\right),
 \end{equation}
 where $J$ is given by \eqref{eq:Jsolve}.
 
 For the disassembly rate, we assume there is both a basal rate of disassembly $d_0$ as well as an additional disassembly which depends on the concentration of depolymerase at the flagellar tip. For simplicity, we will assume that the depolymerase concentration at the $i$\textsuperscript{th} flagellum is proportional to the kinesin concentration $c_{d,i}(L_i)$. (As discussed in \cite{fai2019length}, this would be the case so long as the depolymerase follows the same pattern of ballistic anterograde to diffusive retrograde motion.) Therefore, we have
  \begin{equation}
\text{disassembly rate} = d_0+d_1 c_{d,i}(L_i) = d_0+d_1\frac{J L_i}{D}.
 \end{equation}
 
 Subtracting the assembly and disassembly rates yields the following system of coupled nonlinear ODE's:
 \begin{align}
 \frac{\ud L_1}{\ud t} &= \gamma J \left(T-L_1-L_2\right)-d_0-d_1\frac{J L_1}{D} \\
  \frac{\ud L_2}{\ud t} &= \gamma J \left(T-L_1-L_2\right)-d_0-d_1\frac{J L_2}{D},
 \end{align}
 where $J$ is given by \eqref{eq:Jsolve} above.
 }
 \end{appendices}

\end{document}